\documentclass[11pt]{article}

\usepackage[utf8]{inputenc}
\usepackage[T1]{fontenc}
\usepackage[english]{babel}
\usepackage{amsmath,amssymb,amsthm}
\usepackage{booktabs}
\usepackage{graphicx}
\usepackage{multirow}
\usepackage{array}
\usepackage[margin=1in]{geometry}
\usepackage{tikz}
\usepackage{pgfplots}
\usepackage{algorithm}
\usepackage{algpseudocode}
\usepackage{caption}
\usepackage{subcaption}
\usepackage[numbers,sort&compress]{natbib}
\usepackage[colorlinks=true,linkcolor=blue,citecolor=blue,urlcolor=blue]{hyperref}
\usepackage{cleveref}

\pgfplotsset{compat=1.17}
\usetikzlibrary{positioning,arrows.meta,calc,fit,backgrounds,shapes.geometric,decorations.pathreplacing}

\newtheorem{proposition}{Proposition}
\newtheorem{remark}{Remark}
\newtheorem{definition}{Definition}

\DeclareMathOperator*{\argmax}{arg\,max}
\DeclareMathOperator*{\argmin}{arg\,min}

\newcommand{\Dset}{\mathcal{D}}

\tikzset{
  stage/.style={draw, rounded corners=2pt, align=center, font=\small,
                minimum height=8mm, inner xsep=4pt, fill=black!3},
  gate/.style={draw, rounded corners=2pt, align=center, font=\small,
               minimum height=8mm, fill=black!8},
  nd/.style={draw, rounded corners=2pt, align=center, font=\scriptsize,
             minimum height=7mm, fill=black!5},
  flow/.style={-{Latex[length=2mm]}, thick},
  raw/.style={-{Latex[length=2mm]}, thick, dashed},
  lbl/.style={font=\scriptsize\itshape}
}

\title{\bfseries Self-Explaining Segment Trees: A KPI-Conditioned Segmentation
Framework for Business Analytics with Node-Level Explanation via Recursive
Subspace Partitioning}

\author{%
\normalsize
\begin{tabular}{@{}cc@{}}
Girish G N & Dhanashekar Kandaswamy \\[1mm]
{\small\itshape Head of AI/ML} &
{\small\itshape The Ohio State University} \\
{\small Zeru Finance} &
{\small USA}
\end{tabular}%
}

\date{}

\begin{document}

\maketitle

\begin{abstract}
Business users confronted with a moving metric need to know which part of their
data moved and why. The data-explanation literature answers with predicates:
conjunctions of attribute-value conditions, searched over the lattice of such
conjunctions, that isolate the responsible records. Predicates are exact and
execute directly as filters, but they describe only axis-aligned regions of the
raw attribute space, and no compact conjunction captures a segment defined by a
combination of continuous tendencies rather than by any single attribute value.
This paper presents Self-Explaining Segment Trees, an architecture that adopts
the opposite representation: an explanation is a multivariate cluster located in
a feature subspace selected for relevance to one designated key performance
indicator. The framework selects that subspace once per KPI from Shapley
attributions over a decision-tree surrogate, partitions the population
recursively while choosing the branching factor independently at every node
through mixture-model silhouette search, and attaches to every node---internal
nodes as well as leaves---a dual payload of standardized effect sizes over all
numeric features and type-dependent contribution profiles over user-designated
dimensions, both computed against untransformed data so that the surfaced values
carry the user's own units and category labels. A stance layer reduces any depth
of the tree to its extremal KPI-suppressing and KPI-amplifying segments. We
establish that the tree is a partial hierarchical decomposition whose node count
is bounded by the depth limit and the minimum segment size, that construction
terminates, and that per-tree cost is quadratic in population size in the
degenerate case and geometrically decaying across depth in the balanced case.
This is an architecture and methodology paper. We report no predictive-accuracy
or validation results, and \Cref{sec:limitations} states outcome validation as
future work.
\end{abstract}

\section{Introduction}
\label{sec:intro}

Self-service analytics platforms promise that a non-technical user can answer
their own questions about their own data. Adoption studies report the promise
only partly met: users lack the analytical training and the cognitive budget to
explore data in depth, generate their own analyses, or interpret model output,
and system complexity compounds both gaps \citep{ssbi}. A commercial director
can read that revenue fell. Determining where it fell, and what characterizes
the segment responsible, remains an analyst's job.

The database and data-mining communities have worked this problem for over two
decades under the heading of data explanation. Sarawagi et al.\ introduced
discovery-driven exploration, precomputing exception indicators at every
aggregation level of a cube so anomalies surface rather than being hunted
\citep{sarawagi1998}, and subsequently formalized explaining differences between
multidimensional aggregates \citep{sarawagi1999}. Scorpion works backward from
an outlier in an aggregate query to the predicate over input tuples responsible
for it \citep{scorpion}. Smart drill-down replaces the OLAP drill-down operator
with one surfacing rules that summarize interesting groups of tuples
\citep{smartdrill}. The \textsc{Diff} operator unifies the family into a single
relational primitive, demonstrating that many explanation engines are instances
of one aggregation operator parameterized by a difference metric \citep{diff}.

A single architectural commitment runs through all of them: an explanation is a
predicate. The output is a conjunction of conditions over raw attributes, and
the computational problem is searching that conjunction lattice efficiently. The
commitment buys exactness, direct translation into a filter, and immediate
actionability. It also fixes what can be described. A predicate carves
axis-aligned regions of the raw attribute space and nothing else. A segment
characterized by moderately elevated order frequency, moderately depressed
discount rate, and moderately extended payment cycle---no one of which crosses a
threshold any single condition would isolate---has no compact conjunctive
description. Predicate search either misses it or returns a disjunction of
narrow rules that reads as noise rather than as one insight.

This paper describes an architecture built on the opposite commitment: an
explanation is a cluster. Self-Explaining Segment Trees (SEST) produce a
recursive decomposition in which every node is a multivariate region of a
feature subspace learned for one specific KPI, carrying a machine-generated
description of what distinguishes it. Three decisions separate the architecture
from generic hierarchical clustering. The subspace is selected per KPI rather
than per dataset, so a dataset with three KPIs yields three structurally
distinct trees and the segmentation is conditioned on the question rather than
on the data alone. The branching factor is selected per node rather than
globally, so regions that decompose into different numbers of subgroups are
permitted to do so. Explanation is computed at construction time for every node
against untransformed data, so any node a user reaches during traversal is
immediately readable in the units and category labels of the source table.

The contributions are as follows.

\begin{itemize}

\item \textbf{A cluster-based formulation} of KPI root-cause decomposition
(\Cref{sec:problem}), stated against the predicate-based formulation of prior
work \citep{sarawagi1998,sarawagi1999,scorpion,smartdrill,diff} with the
representational trade-off characterized in both directions
(\Cref{sec:position}).

\item \textbf{The SEST construction} (\Cref{sec:method}): KPI-conditioned
subspace selection over a surrogate model, recursive partial partitioning with
per-node model-order selection, and construction-time dual explanation, given as
algorithms and grounded line by line in the deployed implementation.

\item \textbf{Structural guarantees} (\Cref{sec:properties}).
\Cref{prop:size} bounds node count by the depth limit and minimum segment size
and establishes termination; \Cref{prop:cost} gives the construction cost model
and separates the degenerate and balanced regimes. \Cref{rem:coverage} states
the coverage deficit the partial-partition property admits and which the
guarantees do not bound.

\end{itemize}

This is an architecture and methodology paper. No predictive-accuracy claim,
segmentation-quality measurement, or comparative evaluation appears anywhere in
it. \Cref{sec:limitations} sets out what outcome validation would require.

\section{Related Work}
\label{sec:related}

\subsection{Explaining Aggregates and Data Cubes}

Discovery-driven exploration precomputes exception indicators across all
aggregation levels of a cube, guiding attention by visual cues on cells whose
values depart from expectation \citep{sarawagi1998}. The follow-on work
formalizes explaining differences between two multidimensional aggregates by
locating the cell pairs contributing most to the gap \citep{sarawagi1999}. Both
operate over a pre-declared dimension hierarchy, which fixes the granularity at
which an explanation can be expressed before analysis begins.

Scorpion reformulates the task as a search for predicates whose removal from the
input restores an aggregate to its expected value, giving an explicitly causal
reading to the returned conjunction \citep{scorpion}. Smart drill-down replaces
the standard drill-down operator with one returning rules containing wildcards,
so that a single rule summarizes a group of tuples at a chosen level of
specificity \citep{smartdrill}. The \textsc{Diff} operator generalizes the
family into a relational aggregation primitive parameterized by a difference
metric, and demonstrates that several published explanation engines are
recoverable as instances \citep{diff}.

These systems answer \textit{which records are responsible}. The present work
answers a different question---\textit{what characterizes the population that is
responsible}---and accepts a weaker form of answer in exchange for a wider class
of describable segments.

\subsection{Subgroup Discovery and Exceptional Model Mining}

Subgroup discovery seeks subsets of a dataset, describable by conditions on
attributes, within which a target concept deviates from the population
\citep{sd}. Exceptional model mining generalizes the target from a single
attribute to a model fitted over several attributes, so that a subgroup is
exceptional when the fitted model differs rather than when a mean shifts
\citep{emm}. Both frameworks are target-aware in the sense SEST is: the search
is driven by deviation in a designated quantity rather than by unsupervised
structure alone.

The divergence is representational. Subgroup discovery retains attribute
conditions as its description language and inherits the search-space structure
that follows. SEST abandons that language, which forfeits the exhaustive-search
guarantees and quality-measure theory the subgroup-discovery literature has
developed \citep{sd}, and gains the ability to describe segments no condition
set delimits.

\subsection{Explainable Clustering}

Moshkovitz et al.\ formalize explainable clustering as approximation of an
optimal $k$-means or $k$-medians solution by a decision tree with $k$ leaves,
and prove an $\Omega(\log k)$ lower bound on the cost of that explainability
\citep{exkmeans}; subsequent work tightens the achievable bounds
\citep{exkmeans2}. A recent survey organizes the field into decision-tree,
rule-based, prototype-based, and polyhedral explanation models and documents a
recurring tension in which interpretability gains cost clustering quality
\citep{icsurvey}.

That line explains a clustering that already exists, and the explanation is
again a tree of attribute conditions. SEST inverts the dependency: the target
conditions the clustering before it is computed, and the explanation is an
effect-size ranking over the original feature space rather than a partition
rule. The two are complementary---a decision-tree summary could be fitted to a
SEST node---and \Cref{sec:limitations} treats that combination as open work.

\subsection{Component Techniques}

The constituent methods are established. Shapley additive explanations unify a
family of additive attribution measures and identify the unique solution in that
class satisfying a stated set of properties \citep{shap}; SEST uses them as a
selection filter over features rather than as per-prediction explanations. The
silhouette coefficient scores a partition by comparing within-cluster tightness
against separation from the nearest alternative cluster \citep{silhouette}, and
serves here as the internal criterion for per-node order selection. Ward's
method merges the cluster pair minimizing the increase in within-cluster
variance at each step \citep{ward}, and supplies the partition once the order is
fixed. Implementation rests on standard estimators \citep{sklearn}.

\subsection{Positioning}
\label{sec:position}

\Cref{tab:priorwork} states the representational difference directly. Every
prior family listed expresses explanations in a description language of
attribute conditions. The trade-off is genuine and runs in both directions. A
predicate is exact, reproduces as a filter, and requires no interpretation; a
cluster is approximate, requires reading an effect-size ranking, and cannot be
replayed as a query. Where the responsible population is axis-aligned---which in
transactional business data it frequently is---the predicate formulation is the
better instrument. SEST addresses the residual class, and the exploratory
setting in which no specific outlier has yet been identified and what is wanted
is a navigable decomposition rather than an answer to a posed question.

\begin{table}[t]
\centering
\caption{Explanation representation across prior families and in SEST.}
\label{tab:priorwork}
\begin{tabular}{lll}
\toprule
Approach & Explanation form & Describable segments \\
\midrule
Cube exploration \citep{sarawagi1998}   & Flagged exception cells      & Cells of a declared hierarchy \\
Aggregate differences \citep{sarawagi1999} & Contributing cell pairs   & Cells of a declared hierarchy \\
Scorpion \citep{scorpion}               & Attribute predicate          & Axis-aligned conjunctions \\
Smart drill-down \citep{smartdrill}     & Rules with wildcards         & Axis-aligned conjunctions \\
\textsc{Diff} \citep{diff}              & Predicate sets               & Axis-aligned conjunctions \\
Subgroup discovery \citep{sd,emm}       & Attribute conditions         & Axis-aligned conjunctions \\
Explainable $k$-means \citep{exkmeans}  & Decision tree over clusters  & Axis-aligned, post hoc \\
\midrule
SEST (this work)                        & Cluster in learned subspace  & Arbitrary multivariate regions \\
\bottomrule
\end{tabular}
\end{table}

\section{Methodology}
\label{sec:method}

\Cref{fig:pipeline} gives the end-to-end architecture. The remainder of this
section develops each stage.

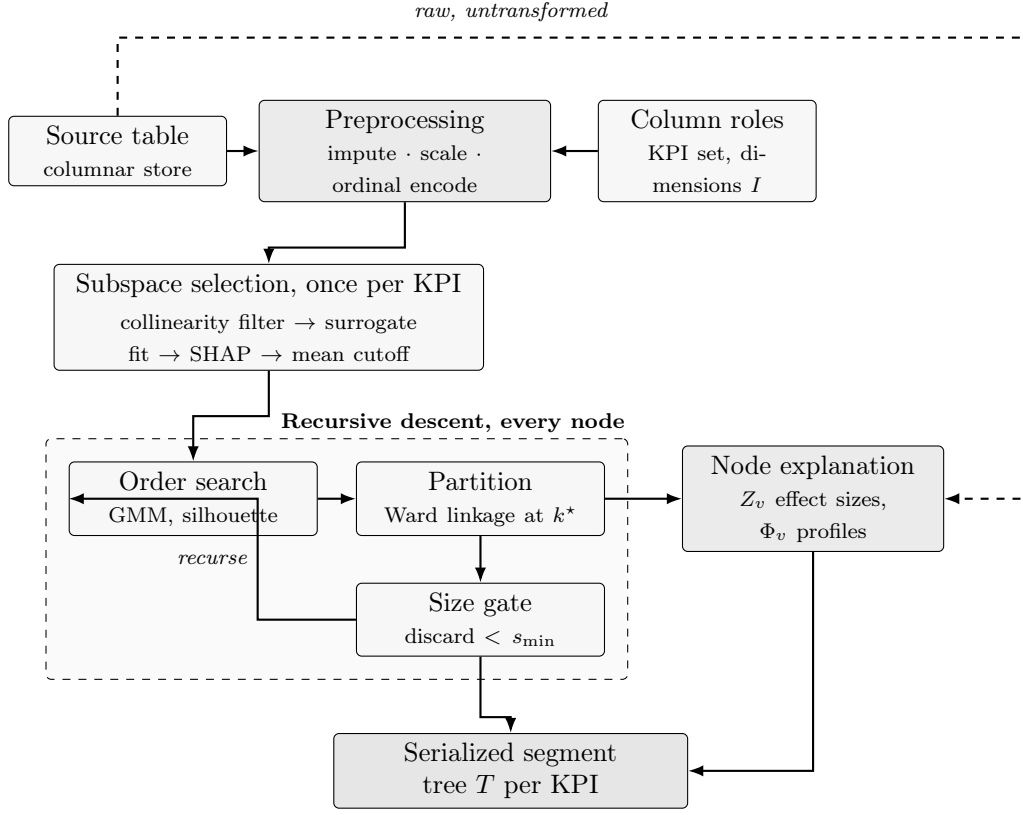
\begin{figure}[t]
\centering
\begin{tikzpicture}[x=1mm,y=1mm]

\node[stage,text width=26mm] (src)  at (0,0)   {Source table\\{\scriptsize columnar store}};
\node[gate, text width=36mm] (prep) at (38,0)
  {Preprocessing\\{\scriptsize impute $\cdot$ scale $\cdot$ ordinal encode}};
\node[stage,text width=26mm] (cfg)  at (78,0)  {Column roles\\{\scriptsize KPI set, dimensions $I$}};
\draw[flow] (src) -- (prep);
\draw[flow] (cfg) -- (prep);

\node[stage,text width=54mm] (sel) at (20,-22)
  {Subspace selection, once per KPI\\[0.6mm]
   {\scriptsize collinearity filter $\to$ surrogate fit $\to$ SHAP $\to$ mean cutoff}};
\draw[flow] (prep.south) -- ++(0,-6) -| (sel.north);

\node[stage,text width=30mm] (ord)  at (10,-46) {Order search\\{\scriptsize GMM, silhouette}};
\node[stage,text width=30mm] (part) at (48,-46) {Partition\\{\scriptsize Ward linkage at $k^\star$}};
\node[stage,text width=30mm] (disc) at (48,-62) {Size gate\\{\scriptsize discard $<s_{\min}$}};
\begin{scope}[on background layer]
\node[draw,dashed,rounded corners=2pt,fit=(ord)(part)(disc),inner sep=3mm,fill=black!2]
  (rec) {};
\end{scope}
\node[font=\scriptsize\bfseries,anchor=south east,yshift=0.4mm,
      fill=white,inner sep=1pt] at (rec.north east)
  {Recursive descent, every node};

\draw[flow] (sel.south) -- ++(0,-6) -| (ord.north);
\draw[flow] (ord) -- (part);
\draw[flow] (part) -- (disc);
\draw[flow] (disc.west) -- ++(-13,0) |- node[lbl,pos=0.25,left,align=right]
  {recurse} (ord.west);

\node[gate,text width=32mm] (expl) at (92,-46)
  {Node explanation\\{\scriptsize $Z_v$ effect sizes, $\Phi_v$ profiles}};
\draw[flow] (part.east) -- (expl.west);
\draw[raw] (src.north) |- (120,15) -- (120,-46) -- (expl.east);
\node[lbl,anchor=south] at (52,15.8) {raw, untransformed};

\node[stage,fill=black!10,text width=44mm] (out) at (52,-82)
  {Serialized segment tree $T$ per KPI};
\draw[flow] (disc.south) -- ++(0,-8) -| (out.north);
\draw[flow] (expl.south) |- (out.east);

\end{tikzpicture}
\caption{SEST construction pipeline. Subspace selection runs once per KPI;
order selection, partitioning, and the size gate recur at every node. The dashed
edge carries untransformed source values to the explanation stage, so no
user-facing quantity is expressed in scaled or ordinal-encoded coordinates.}
\label{fig:pipeline}
\end{figure}

\subsection{Data Model and Preprocessing}
\label{sec:data}

Let $\Dset = \{x_1,\dots,x_n\}$ be a table over a mixed attribute space
$A_1 \times \cdots \times A_m$, each $A_j$ numeric or categorical. One numeric
attribute $y$ is designated the KPI. A set $I$ of dimensions of interest is
designated by the user and is profiled at every node regardless of its
statistical relevance to $y$. Attributes the user marks as irrelevant are
dropped before any computation.

Numeric attributes are mean-imputed and min--max scaled; categorical attributes
are mode-imputed and ordinal-encoded, through a single column-wise transformer.
The encoded table serves one purpose, distance computation during partitioning,
and the original table is retained alongside it. Every quantity that reaches the
user is computed from the original table. Ordinal encoding of nominal attributes
imposes an artificial order and therefore an artificial distance geometry; the
choice reflects the deployed implementation rather than a defended modeling
position, and \Cref{sec:limitations} records it.

\subsection{Problem Statement}
\label{sec:problem}

SEST constructs a rooted tree $T=(V,E)$ with root $v_0$, each node $v$
associated with a row subset $\Dset_v \subseteq \Dset$, subject to four
conditions.

\begin{definition}[Segment tree]
\label{def:tree}
$T$ is a \emph{segment tree} for $(\Dset, y, I)$ when:
\textup{(i)} $\Dset_{v_0} = \Dset$;
\textup{(ii)} for every internal node $v$ with child set $C(v)$, the children are
pairwise disjoint and contained in the parent,
\begin{equation}
\Dset_a \cap \Dset_b = \emptyset \quad \forall a \neq b \in C(v),
\qquad
\bigcup_{c \in C(v)} \Dset_c \subseteq \Dset_v ;
\label{eq:partial}
\end{equation}
\textup{(iii)} $v$ is a leaf whenever $\mathrm{level}(v) \geq d_{\max}$ or
$|\Dset_v| < s_{\min}$; and
\textup{(iv)} every node carries an explanation $E_v = (Z_v, \Phi_v)$ as
constructed in \Cref{sec:explain}.
\end{definition}

The containment in \Cref{eq:partial} is proper in general. Children falling
below $s_{\min}$ are discarded rather than retained, on the position that a
segment too small to characterize reliably is worse than no segment. $T$ is
therefore a partial hierarchical decomposition, and row-level accounting across
a level must account for the discarded remainder. \Cref{rem:coverage} returns to
this.

A fifth condition distinguishes the construction from unsupervised hierarchical
clustering. The subspace within which each partition is computed is a function
of $y$: there exists a selection operator $S$ such that partitioning at every
node operates on $\pi_{S(\Dset,y)}(\Dset_v)$, the projection onto the selected
subspace. $T$ therefore depends on $y$, and distinct KPIs over one table induce
structurally distinct trees.

\subsection{KPI-Conditioned Subspace Selection}
\label{sec:select}

\Cref{alg:select} states the selection procedure. A collinearity filter removes
one member of each strongly correlated numeric pair. A decision-tree surrogate
is fitted to predict $y$, with the task type inferred from the KPI's own
distribution: regression when $y$ is numeric and carries more than a small fixed
number of distinct values, classification otherwise. Shapley attributions over
the fitted surrogate \citep{shap} rank the surviving features, and those at or
above the mean attribution are retained. The user's dimensions of interest are
unioned in unconditionally, which guarantees $\Phi_v$ is populated for the
columns the user cares about even where those columns carry no signal for $y$.

\begin{algorithm}[t]
\caption{KPI-conditioned subspace selection}
\label{alg:select}
\begin{algorithmic}[1]
\Require encoded table $\Dset'$; KPI column $y$; dimensions of interest $I$;
correlation bound $\rho$; attribution sample size $N_\phi$
\Ensure feature subspace $F \subseteq \mathrm{cols}(\Dset') \setminus \{y\}$
\State $X \gets \Dset' \setminus \{y\}$
\State $R \gets \{f : \exists\, g \neq f \text{ numeric},\;
        |\mathrm{corr}(f,g)| \geq \rho\}$
  \Comment{collinear redundancy}
\State $X \gets X \setminus R$
\State $\textit{reg} \gets \mathrm{numeric}(y) \wedge |\mathrm{unique}(y)| > \nu$
  \Comment{task type from the KPI itself}
\State $M \gets \textit{reg}\;?\;\textsc{TreeRegressor}()\;:\;\textsc{TreeClassifier}()$
\State $X_s \gets \mathrm{sample}\big(X,\ \min(N_\phi, |X|)\big)$
  \Comment{caps attribution cost independently of $n$}
\State $M.\textsc{Fit}\big(X_s,\; y[\mathrm{index}(X_s)]\big)$
\State $\phi \gets \textsc{ShapAttributions}(M, X_s)$
  \Comment{mean $|\phi_f|$ per feature}
\State $F \gets \{f \in X : \phi_f \geq \mathrm{mean}(\phi)\}$
\State $F \gets F \cup I$
  \Comment{user dimensions retained unconditionally}
\State \Return $F$
\end{algorithmic}
\end{algorithm}

Two properties of \Cref{alg:select} carry consequences downstream. Line~6 bounds
attribution cost at a constant independent of $n$, which keeps selection off the
critical path as the table grows and is the reason selection does not appear in
the cost model of \Cref{prop:cost}. Line~9 applies a mean-attribution cutoff, a
threshold chosen for its simplicity rather than derived from a test of
significance; no permutation test or false-discovery control is applied.

The procedure runs once per KPI, at the root, and $F$ is then fixed for the
entire tree. A feature discriminative only within an already-isolated
subpopulation cannot enter the subspace at any depth. This is the most
consequential structural limitation of the design and \Cref{sec:limitations}
treats it as such.

\subsection{Recursive Partitioning with Per-Node Order Selection}
\label{sec:partition}

\Cref{alg:build} states the descent. At each node the branching factor is chosen
by fitting a Gaussian mixture for every candidate order in a fixed range and
scoring the resulting assignment by mean silhouette \citep{silhouette}. The
selected order then parameterizes a separate Ward-linkage agglomerative fit
\citep{ward}, whose labels define the actual partition. Children below $s_{\min}$
are discarded at line~13, which is where the partial-partition property of
\Cref{eq:partial} originates.

\begin{algorithm}[t]
\caption{Segment tree construction, \textsc{Build}$(\Omega, \ell, p)$}
\label{alg:build}
\begin{algorithmic}[1]
\Require row indices $\Omega$; depth $\ell$; path $p$; subspace $F$ from
\Cref{alg:select}
\Ensure node $v$ rooted at $\Omega$
\State $v \gets \textsc{Node}(\ell,\, \Omega,\, p)$
\State $v.Z \gets \textsc{EffectSizes}(\Omega)$
  \Comment{\Cref{eq:z}, computed on the raw table}
\State $v.\Phi \gets \textsc{ContributionProfile}(\Omega, I)$
  \Comment{raw table}
\If{$\ell \geq d_{\max}$ \textbf{or} $|\Omega| < s_{\min}$}
  \State \Return $v$
    \Comment{leaf by \Cref{def:tree}(iii)}
\EndIf
\State $S \gets \textit{encoded}[\Omega,\, F]$
\State $k^\star \gets \displaystyle\argmax_{k \in [2,\, k_{\max}]}
       \ \mathrm{sil}\big(S,\ \textsc{Gmm}(k).\textsc{FitPredict}(S)\big)$
\State $L \gets \textsc{Ward}(k^\star).\textsc{FitPredict}(S)$
  \Comment{order and partition from different model families}
\State $v.\textit{score} \gets \mathrm{sil}(S, L)$
\For{$c \in \{0,\dots,k^\star - 1\}$}
  \State $\Omega_c \gets \{i \in \Omega : L_i = c\}$
  \If{$|\Omega_c| \geq s_{\min}$}
    \State $v.\textit{children}.\textsc{Append}\big(\textsc{Build}(\Omega_c,\,
           \ell+1,\, p \oplus c)\big)$
  \EndIf
    \Comment{else discarded; source of \Cref{eq:partial}}
\EndFor
\State \Return $v$
\end{algorithmic}
\end{algorithm}

\begin{remark}[Order selection and partition are separate fits]
\label{rem:asymmetry}
Line~8 selects $k^\star$ from a mixture model; line~9 partitions with Ward
linkage at that order. The model choosing the branching factor is not the model
producing the partition, and consequently $v.\textit{score}$ recorded at line~10
is computed on the Ward labels rather than on the mixture assignment that
justified the choice of $k^\star$. The two need not agree. The deployed
rationale is that full-covariance mixtures give a more sensitive order signal
while Ward linkage yields more compact partitions for downstream reading, and
the design is recorded here as implemented rather than as validated.
\end{remark}

\subsection{Node-Level Explanation}
\label{sec:explain}

Explanation is not a stage applied to leaves after construction. Every node
computes its payload at the moment it is created, from the original table.

The first component ranks numeric features by standardized effect size against
the population. For node $v$ and numeric feature $f$, writing $\mu_f$ and
$\sigma_f$ for the population mean and standard deviation,
\begin{equation}
z_f(v) \;=\; \frac{\mu_f(\Dset_v) - \mu_f(\Dset)}{\sigma_f(\Dset)},
\qquad
Z_v \;=\; \big\langle f : |z_f(v)| \text{ descending} \big\rangle .
\label{eq:z}
\end{equation}
The ranking covers every numeric column, not only those in $F$. A feature
excluded from the partitioning subspace may still describe the resulting segment
usefully, and the separation of the two roles is deliberate. The reference in
\Cref{eq:z} is the global population at every depth, so a node at depth four
reports deviation relative to the whole table rather than relative to its
parent, and the reported quantity mixes deviation inherited from ancestors with
deviation the node's own split produced. \Cref{sec:limitations} records the
consequence.

The second component profiles the user's dimensions of interest, with the form
determined by column type. Categorical dimensions yield the mode, the leading
and trailing categories by frequency within the segment, and the segment's
coverage of the population's category set. Continuous numeric dimensions yield
the segment mean and sum as percentages of the population's. Numeric dimensions
whose distinct-value count falls below a fixed threshold are treated as ordinal
and receive both forms. Because all three are computed on the untransformed
table, the surfaced values carry the source system's units and category labels.

\subsection{Stance Classification}
\label{sec:stance}

Any depth of the tree admits a reduction to the two segments with the largest
signed KPI deviation. For node $v$ at depth $\ell$ with parent $p$, and a
threshold ratio $\theta$,
\begin{equation}
\delta(v) \;=\; \mu_y(\Dset_v) - \mu_y(\Dset_p),
\qquad
\tau \;=\; \theta\,\big|\mu_y(\Dset_p)\big| ,
\label{eq:stance}
\end{equation}
with $v$ labelled \textsc{Suppressing} when $\delta(v) < -\tau$,
\textsc{Amplifying} when $\delta(v) > \tau$, and \textsc{Neutral} otherwise. The
layer returns $\argmin_v \delta$ and $\argmax_v \delta$ over the two non-neutral
sets.

The classification is a first-order screen. It compares means against a
proportional threshold with no significance test, no adjustment for unequal
segment variance, and no correction for the multiple comparisons implicit in
scanning every node at a depth. The deployed implementation labels the two
classes with terms whose intuitive reading is inverted relative to their
definition; the terminology in \Cref{eq:stance} is the corrected form and
differs from the source.

\subsection{Persistence and Interactive Traversal}
\label{sec:persist}

A node serializes to the record of \Cref{tab:payload}. Membership is stored as
row indices into an immutable source table, so any segment's rows are recoverable
without recomputation and no derived copy of the data is duplicated per node.
Trees are persisted whole, one per KPI.

\begin{table}[t]
\centering
\caption{Serialized node payload. Explanation fields are populated at
construction time for internal nodes and leaves alike.}
\label{tab:payload}
\begin{tabular}{lll}
\toprule
Field & Type & Role \\
\midrule
identifier      & stable key           & Node identity across sessions \\
depth           & integer              & Level in $T$ \\
indices         & integer list         & Row membership into the source table \\
path            & label sequence       & Route from the root \\
score           & real                 & Silhouette of this node's own split \\
$Z_v$           & ranked feature list  & Standardized effect sizes, \Cref{eq:z} \\
$\Phi_v$        & per-dimension record & Type-dependent contribution profile \\
KPI             & column name          & The $y$ this tree was built for \\
children        & node list            & Recursive structure \\
\bottomrule
\end{tabular}
\end{table}

The user's descent is persisted separately as an ordered traversal with a
selection index, so exploration survives across sessions and can be resumed or
replayed. That record is what distinguishes the artifact from a static report:
the tree is navigated rather than read.

\section{Architecture and System Properties}
\label{sec:properties}

\subsection{Structural Guarantees}

\begin{proposition}[Termination and node-count bound]
\label{prop:size}
Let $T$ be a segment tree per \Cref{def:tree} over $n$ rows with depth limit
$d_{\max}$ and minimum segment size $s_{\min} \geq 1$. Construction terminates,
the nodes at any depth $\ell \geq 1$ number at most $\lfloor n / s_{\min}
\rfloor$, and
\[
|V| \;\leq\; 1 + d_{\max}\left\lfloor \frac{n}{s_{\min}} \right\rfloor .
\]
\end{proposition}

\begin{proof}
Fix $\ell \geq 1$. Nodes at depth $\ell$ are pairwise disjoint: siblings are
disjoint by \Cref{eq:partial}, and nodes under distinct parents are disjoint
because each is contained in its parent and the parents are themselves disjoint
by induction on $\ell$, with the base case $\ell=1$ given by \Cref{eq:partial}
at the root. Every node at depth $\ell \geq 1$ satisfies $|\Dset_v| \geq s_{\min}$,
since line~13 of \Cref{alg:build} admits a child only under that condition.
Disjoint subsets of $\Dset$ each of size at least $s_{\min}$ number at most
$\lfloor n/s_{\min} \rfloor$. Summing over $\ell = 1,\dots,d_{\max}$ and adding
the root gives the stated bound. Termination follows because
\Cref{def:tree}(iii) makes every node at depth $d_{\max}$ a leaf and the
recursion increments depth at each call.
\end{proof}

\Cref{prop:size} is an engineering bound before it is a mathematical one. The
serialized payload of \Cref{tab:payload} embeds an explanation at every node, so
document size scales with $|V|$ and with $|I|$, and the proposition is what
makes that size predictable from configuration rather than from data. With the
deployed defaults of \Cref{tab:params} the bound is $1 + 5\lfloor n/10 \rfloor$,
which is linear in $n$ and dominated in practice by the fact that real branching
factors fall well below $k_{\max}$.

\begin{remark}[The coverage deficit is not bounded]
\label{rem:coverage}
\Cref{prop:size} bounds node count but says nothing about coverage. Nothing in
the construction prevents every child at a node from falling below $s_{\min}$
and being discarded, in which case that node's entire population is absent from
all deeper levels while the node itself remains internal by depth. The fraction
of rows surviving to depth $\ell$ is non-increasing in $\ell$, and the deployed
implementation neither reports nor bounds it. A tree that silently drops a third
of its rows between depth one and depth four supports very different conclusions
from one that drops a fiftieth, and the two are not currently distinguishable
from the output.
\end{remark}

\subsection{Cost Model}

\begin{proposition}[Construction cost]
\label{prop:cost}
Let $K = k_{\max}$, $d = d_{\max}$, and $m = |F|$. Order selection at a node of
size $n_v$ costs $\Theta(K n_v^2 m)$, dominated by the silhouette evaluations.
Writing $V_\ell$ for the nodes at depth $\ell$, total construction cost is
\[
\Theta\!\left(K m \sum_{\ell=0}^{d-1} \sum_{v \in V_\ell} n_v^2\right),
\]
which is $\Theta(K m d\, n^2)$ when some level is dominated by a single node,
and $\Theta(K m n^2 \sum_{\ell} b^{-\ell})$ under a uniform branching factor $b$,
a geometric series bounded by $Kmn^2 \cdot b/(b-1)$.
\end{proposition}

\begin{proof}
At a node of size $n_v$, order selection fits $K-1$ mixtures and evaluates $K-1$
silhouettes. A silhouette evaluation requires all pairwise distances over $n_v$
points in $m$ dimensions, hence $\Theta(n_v^2 m)$, which dominates the mixture
fits. Summing over nodes gives the stated total. For the degenerate case,
suppose at each level one node retains a constant fraction of the rows; that
node contributes $\Theta(n^2)$ to its level and there are $d$ levels. For the
balanced case with uniform branching $b$, level $\ell$ holds $b^\ell$ nodes of
size $n/b^\ell$, contributing $b^\ell (n/b^\ell)^2 = n^2 b^{-\ell}$; summing the
geometric series over $\ell \geq 0$ gives $n^2 b/(b-1)$.
\end{proof}

The two regimes differ by a factor of $d$ and diverge sharply with branching, as
\Cref{fig:cost} shows. The quadratic term is the binding constraint in both, and
it originates entirely in exact silhouette evaluation. Subsampling the
silhouette computation, substituting a linear-time validity index, or selecting
the order from the mixture's own likelihood criterion would each remove the
pairwise-distance term; none is implemented.

\begin{figure}[t]
\centering
\begin{tikzpicture}
\begin{axis}[
  width=0.74\textwidth, height=5.6cm,
  xlabel={depth $\ell$},
  ylabel={level cost $/\ Kmn^2$},
  xmin=0, xmax=5, ymin=0.0002, ymax=2,
  ymode=log,
  xtick={0,1,2,3,4,5},
  tick label style={font=\scriptsize}, label style={font=\small},
  legend style={font=\scriptsize, at={(0.03,0.04)}, anchor=south west,
                draw=gray!40, fill=white, row sep=1pt},
  grid=major, grid style={gray!18},
  axis lines=left,
]
\addplot[thick, black, dashed, mark=square, mark size=1.2pt]
  coordinates {(0,1)(1,1)(2,1)(3,1)(4,1)(5,1)};
\addlegendentry{degenerate: one dominant node per level}
\addplot[thick, black, mark=*, mark size=1.2pt]
  coordinates {(0,1)(1,0.5)(2,0.25)(3,0.125)(4,0.0625)(5,0.03125)};
\addlegendentry{balanced, branching $b=2$}
\addplot[thick, black, mark=triangle*, mark size=1.4pt, densely dotted]
  coordinates {(0,1)(1,0.2)(2,0.04)(3,0.008)(4,0.0016)(5,0.00032)};
\addlegendentry{balanced, branching $b=5$}
\end{axis}
\end{tikzpicture}
\caption{Level-wise construction cost from \Cref{prop:cost}, normalized by
$Kmn^2$. Curves are the analytic model evaluated at $b \in \{2,5\}$ and at the
degenerate case; no measurement is implied. Balanced descent decays
geometrically, so total cost is dominated by the root; degenerate descent pays
the root cost at every level.}
\label{fig:cost}
\end{figure}
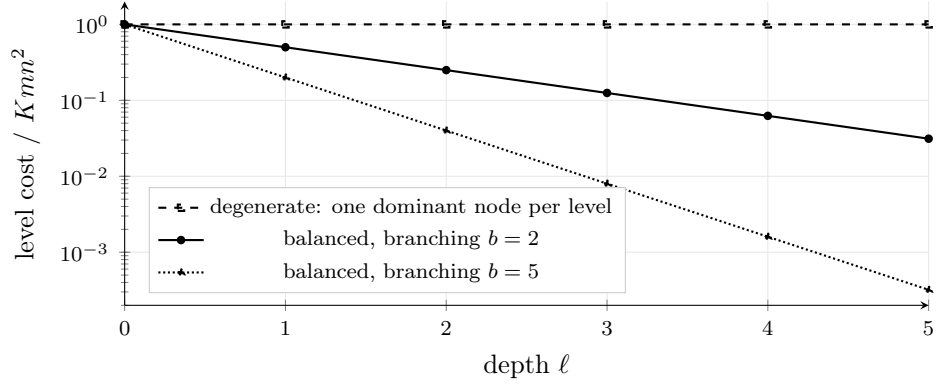

\subsection{Output Structure}

\Cref{fig:tree} shows the shape of a constructed tree and where each payload
component attaches.

\begin{figure}[t]
\centering
\begin{tikzpicture}[x=1mm,y=1mm]

\node[nd,text width=44mm,fill=black!10] (r) at (40,0)
  {\textbf{root}\quad $\Dset_{v_0}=\Dset$\\
   {\scriptsize $Z_{v_0}\equiv 0$ by \Cref{eq:z}; $\Phi_{v_0}$ = population profile}};

\node[nd,text width=31mm] (a) at (0,-26)
  {segment A\\{\scriptsize $Z_A$: $f_3\,{+}1.6$, $f_7\,{-}1.2$}\\
   {\scriptsize\itshape Neutral}};
\node[nd,text width=31mm,fill=black!12] (b) at (40,-26)
  {segment B\\{\scriptsize $Z_B$: $f_1\,{+}1.8$, $f_4\,{+}1.4$}\\
   {\scriptsize\itshape Suppressing: $\argmin_v \delta(v)$}};
\node[nd,text width=31mm] (c) at (80,-26)
  {segment C\\{\scriptsize $Z_C$: $f_2\,{-}1.9$}\\
   {\scriptsize\itshape Amplifying: $\argmax_v \delta(v)$}};

\node[nd,text width=20mm] (b1) at (25,-50) {B.1};
\node[nd,text width=20mm] (b2) at (55,-50) {B.2};

\draw[flow] (r) -- (a);
\draw[flow] (r) -- (b);
\draw[flow] (r) -- (c);
\draw[flow] (b.south) -- (b1.north);
\draw[flow] (b.south) -- (b2.north);

\end{tikzpicture}
\caption{Segment tree structure. Effect-size rankings and contribution profiles
attach to internal nodes as well as leaves, so traversal halts at no
uninterpretable node. Feature labels and magnitudes shown are schematic
illustrations of payload structure and are not measurements.}
\label{fig:tree}
\end{figure}
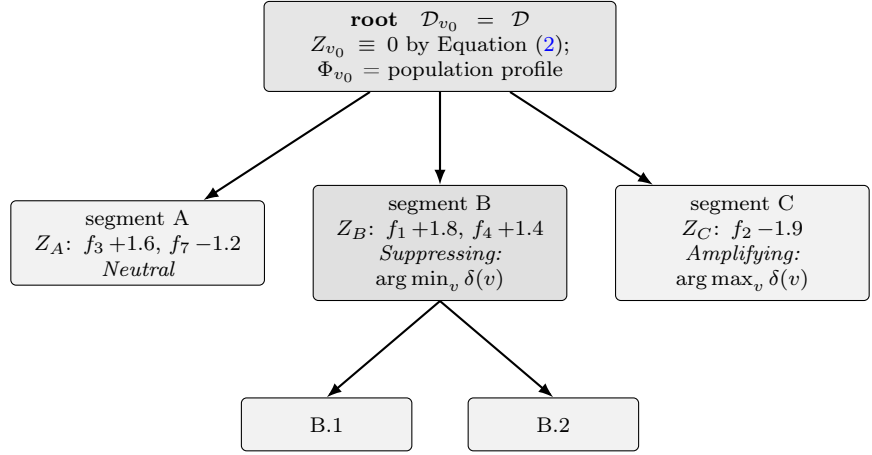

\subsection{Implementation and Configuration}
\label{sec:impl}

The framework is implemented in Python. Estimators for the mixture fits,
agglomerative partitioning, surrogate models, and validity metrics come from
scikit-learn \citep{sklearn}; attributions use a Shapley implementation
\citep{shap}. Trees for distinct KPIs share no state and are constructed
concurrently. Construction is triggered synchronously by an API request and
completes in-request, with no memoization between requests: identical inputs
rebuild the tree in full.

\Cref{tab:params} lists the governing parameters and their deployed values.
Every one is fixed at engine construction and identical across tenants,
datasets, and KPIs. None is exposed for tuning and none is derived from a
calibration procedure.

\begin{table}[t]
\centering
\caption{Governing parameters and deployed values. None is exposed for
per-dataset tuning; none is empirically calibrated.}
\label{tab:params}
\begin{tabular}{llcl}
\toprule
Parameter & Symbol & Value & Role \\
\midrule
Depth limit            & $d_{\max}$ & 5    & Termination; bounds $|V|$ via \Cref{prop:size} \\
Minimum segment size   & $s_{\min}$ & 10   & Termination; source of \Cref{eq:partial} \\
Order-search width     & $k_{\max}$ & 10   & Candidate branching factors \\
Collinearity bound     & $\rho$     & 0.8  & Redundancy filter, \Cref{alg:select} \\
Attribution sample     & $N_\phi$   & 1000 & Caps selection cost in $n$ \\
Task-type cutoff       & $\nu$      & 10   & Regression versus classification surrogate \\
Ordinal cutoff         & ---        & 25   & Selects the form of $\Phi_v$ \\
Stance threshold       & $\theta$   & 0.02 & \Cref{eq:stance} \\
\bottomrule
\end{tabular}
\end{table}

Determinism holds under fixed configuration. Mixture initialization, surrogate
fitting, and attribution sampling are seeded, and Ward linkage is deterministic
given its input, so repeated construction over an unchanged table reproduces an
identical tree. That property is a precondition for the validation work of
\Cref{sec:limitations} rather than a result in itself.

\section{Limitations and Future Work}
\label{sec:limitations}

\paragraph{No outcome validation.} The framework has not been evaluated against
any measure of segmentation quality, explanation fidelity, or utility to an
analyst, and this paper makes no such claim. Validation would require, at
minimum: internal validity indices aggregated across nodes and reported by
depth; recovery measured against synthetic populations with planted structure;
a held-out test of whether a node's top-ranked effect-size features predict
membership better than a random-feature control; and a blinded rating exercise
in which domain experts judge node explanations for correctness and
actionability. The implementation computes a silhouette score at every node and
discards it without aggregation, so the cheapest of these requires instrumentation
rather than new algorithmic work.

\paragraph{The discriminating comparison is unrun.} \Cref{sec:position} claims
that cluster-based explanation describes segments predicate search cannot. The
experiment that would settle it constructs a population containing both
axis-aligned segments and segments defined by combinations of continuous
tendencies, and measures recovery of each class separately for SEST and for a
predicate-based baseline \citep{scorpion,diff}. Until that is run, the
positioning of \Cref{tab:priorwork} rests on the representational argument
alone.

\paragraph{Subspace selection does not recur.} \Cref{alg:select} runs once per
KPI at the root and fixes $F$ for the entire tree, so a feature discriminative
only inside an already-isolated subpopulation is unreachable at every depth.
This is precisely the situation deep traversal exists to expose. Re-selecting
per node costs one attribution pass per node, and by line~6 of \Cref{alg:select}
that pass is $O(1)$ in $n$, so the change is affordable; its effect on
explanation quality at depth is unmeasured.

\paragraph{Effect sizes reference the population, not the parent.} \Cref{eq:z}
standardizes against $\Dset$ at every depth. A node at depth four therefore
reports a mixture of deviation inherited from its ancestors and deviation its
own split produced, and the two are not separable from the output. A
parent-referenced variant is a one-line change with a large effect on how deep
nodes read.

\paragraph{Coverage is unreported.} By \Cref{rem:coverage} the row fraction
surviving to a given depth is neither bounded by the construction nor recorded
in the output. Reporting it per level is straightforward and would materially
change how leaf-level statements should be interpreted.

\paragraph{Configuration is uncalibrated.} The eight parameters of
\Cref{tab:params} are fixed constants applied identically to every dataset. Each
is a clean ablation axis, and none has been swept. The mean-attribution cutoff
at line~9 of \Cref{alg:select} and the order-selection asymmetry of
\Cref{rem:asymmetry} are the two whose justification is weakest.

\paragraph{Statistical informality of the stance layer.} \Cref{eq:stance}
screens on a mean difference against a proportional threshold, applied across
every node at a depth without correction for multiplicity. False positives
should be expected at a rate growing with the width of the level. The
quality-measure theory developed in the subgroup-discovery literature
\citep{sd,emm} offers principled replacements.

\paragraph{Encoding of nominal attributes.} Ordinal encoding imposes an order
and a distance geometry on attributes that carry neither
(\Cref{sec:data}). A mixed-type distance would remove the artifact; the cost to
current output is unquantified.

\paragraph{Recomputation and freshness.} Construction rebuilds the full tree per
request with no memoization or incremental maintenance, which trades repeated
quadratic work for operational simplicity and guaranteed freshness. Incremental
maintenance under append-only updates is open.

\paragraph{Composition with predicate methods.} The most promising extension the
trade-off of \Cref{sec:position} suggests is a hybrid: cluster to discover a
segment, then fit a compact predicate to describe it, recovering exactness and
executability while retaining the wider class of discoverable segments. Nothing
in the present architecture forecloses it.

\section{Conclusion}
\label{sec:conclusion}

We presented Self-Explaining Segment Trees, an architecture for KPI root-cause
decomposition that represents an explanation as a multivariate cluster in a
learned subspace rather than as a predicate over raw attributes. The framework
selects that subspace once per KPI from Shapley attributions over a surrogate
model, partitions recursively while choosing the branching factor independently
at every node, discards segments too small to characterize, and attaches
standardized effect sizes and type-dependent contribution profiles to every node
at construction time, computed against untransformed data so that surfaced
values carry the source system's own units and labels. A stance layer reduces
any depth to its extremal KPI-suppressing and KPI-amplifying segments, and a
persisted traversal record makes the tree navigable across sessions.

The structural results characterize what the construction guarantees.
\Cref{prop:size} bounds node count by the depth limit and the minimum segment
size and establishes termination, which makes serialized size predictable from
configuration rather than from data. \Cref{prop:cost} locates the cost in exact
silhouette evaluation and separates a degenerate regime paying the root cost at
every level from a balanced regime whose cost decays geometrically with depth.
\Cref{rem:coverage} states what the guarantees do not cover: the construction
admits a coverage deficit it neither bounds nor reports.

The open question is representational and empirical at once. In transactional
business data, how large is the class of KPI-driving segments that no compact
attribute predicate delimits? Where that class is small the predicate
formulation is sufficient and the architecture described here is redundant.
Where it is substantial, cluster-based explanation reaches segments a mature
literature cannot express. \Cref{sec:limitations} specifies the comparison that
decides it, and states the outcome validation this work leaves open.


\end{document}